\documentclass[12pt]{article} 
\usepackage[top=3cm,bottom=3cm,left=3.5cm,right=3.5cm]{geometry}
\usepackage{mathptmx}
\usepackage{mathrsfs}
\usepackage{amsmath}
\usepackage{bm}
\usepackage{fancyhdr}
\usepackage[colorlinks,
citecolor=blue,linkcolor=red]{hyperref}
\usepackage{amsfonts}
\usepackage{natbib}
\usepackage{amsthm}
\usepackage{amssymb}
\usepackage{color}
\usepackage{lscape}
\usepackage{rotating}
\usepackage{subfigure}
\usepackage{graphicx}
\usepackage{placeins}
\usepackage{multirow}
\usepackage{tabularx,booktabs,caption}
\newcolumntype{C}[1]{>{\Centering}m{#1}}

\makeatletter
\AtBeginDocument{%
  \expandafter\renewcommand\expandafter\subsection\expandafter
    {\expandafter\@fb@secFB\subsection}%
  \newcommand\@fb@secFB{\FloatBarrier
    \gdef\@fb@afterHHook{\@fb@topbarrier \gdef\@fb@afterHHook{}}}%
  \g@addto@macro\@afterheading{\@fb@afterHHook}%
  \gdef\@fb@afterHHook{}%
}
\makeatother

\renewcommand{\hat}{\widehat}

\newcommand{\NB}{\mathrm{NB}}
\newcommand{\BHerm}{\mathrm{BHerm}}

\newcommand{\BVB}{\mathrm{BVB}}
\newcommand{\BNB}{\mathrm{BNB}}
\newcommand{\Bin}{\mathrm{Bin}}
\newcommand{\Poi}{\mathrm{Poi}}
\newcommand{\BPoi}{\mathrm{BPoi}}
\newcommand{\E}{\mathbb{E}}
\newcommand{\V}{\mathbb{V}}
\newcommand{\R}{\mathbb{R}}
\newcommand{\N}{\mathbb{N}}
\newcommand{\Z}{\mathbb{Z}}

\newcommand{\fX}{\mathbf{X}}
\newcommand{\fZ}{\mathbf{Z}}

\usepackage{enumerate}
\usepackage{float}
\usepackage{amssymb}

\renewcommand{\hat}{\widehat}

\newtheorem{theorem}{Theorem}

\theoremstyle{definition}

\newtheorem{rmk}{Remark}

\renewcommand{\P}{\mathbb{P}}
\renewcommand{\R}{\mathbb{R}}

\newcommand{\ba}{\begin{equation}}
\newcommand{\ea}{\end{equation}}

\usepackage{epstopdf}
\usepackage{enumitem}
\usepackage{multirow}
\usepackage{booktabs}

\usepackage{dsfont}
\newcommand{\indfct}{\mathds{1}}

\begin{document}

\title{Novel Stein-type Characterizations of Bivariate Count Distributions with Applications}

\author{Shaochen Wang \thanks{School of Mathematics, South China University of Technology, Guangzhou, P.\ R.\ China, E-Mail:
\href{mailto:mascwang@scut.edu.cn}{\nolinkurl{mascwang@scut.edu.cn}}. ORCID:\href{https://orcid.org/0000-0003-1389-320X}{\nolinkurl{0000-0003-1389-320X}}.}
\and
Christian H.\ Wei\ss{}\thanks{Department of Mathematics and Statistics, Helmut Schmidt University, 22043 Hamburg, Germany. E-Mail: \href{mailto:weissc@hsu-hh.de}{\nolinkurl{weissc@hsu-hh.de}}. ORCID: \href{https://orcid.org/0000-0001-8739-6631}{\nolinkurl{0000-0001-8739-6631}}.}
}

\date{\today}
\maketitle

\vspace{-1cm}

\begin{abstract}
The derivation and application of Stein identities have received considerable research interest in recent years, especially for continuous or discrete-univariate distributions. In this paper, we complement the existing literature by deriving and investigating Stein-type characterizations for the three most common types of bivariate count distributions, namely the bivariate Poisson, binomial, and negative-binomial distribution. Then, we demonstrate the practical relevance of these novel Stein identities by a couple of applications, namely the deduction of sophisticated moment expressions, of flexible goodness-of-fit tests, and of novel tests for the symmetry of bivariate count distributions. The paper concludes with an analysis of real-world data examples.
\end{abstract}

\begin{quote}
\end{quote}

\begin{quote}
\end{quote}

\thispagestyle{empty}
\pagenumbering{gobble}

\pagenumbering{arabic}

\setcounter{page}{1}

\section{Introduction}
\label{Introduction}
The idea to uniquely characterize a distribution family by certain moment identities dates back to \citet{stein72,stein86}, and one thus speaks of ``Stein's method'' and ``Stein identities'' in this regard. Stein identities have been developed for various distribution families, see \citet{sudheesh09,sudheesh12,landsman16,wang2023} for examples and references. In particular, such Stein identities have been derived for several \emph{univariate} count distributions, i.e., referring to some discrete-quantitative random variable~$X$ with range contained in $\mathbb{N}_0 = \{0, 1, \ldots\}$.
Probably the most relevant cases for practice are, see \citet{sudheesh12}:
\begin{description}
	\item[Poisson (Poi) distribution:] $X\sim\Poi(\mu)$ with $\mu>0$ if and only if
\ba
\label{SteinChen}
\E\big[X\, f(X)\big]\ =\ \mu\, \E\big[f(X+1)\big]
\ea
holds for all bounded functions $f:\N_0\to\R$;
	\item[binomial (Bin) distribution:] $X\sim\Bin(N,\pi)$ with $N\in\N=\{1,2,\ldots\}$ and $\pi\in (0,1)$ if and only if
\ba
\label{SteinBin}
(1-\pi)\,\E\big[X\,f(X)\big] = \pi\,\E\big[(N-X)\,f(X+1)\big]
\ea
holds for all bounded functions $f:\N_0\to\R$;
	\item[negative-binomial (NB) distribution:] $X\sim\NB(\nu,\pi)$ with $\nu>0$ and $\pi\in (0,1)$ if and only if
\ba
\label{SteinNB}
\E\big[X\,f(X)\big]\ =\ (1-\pi)\, \E\big[(\nu+X)\,f(X+1)\big]
\ea
holds for all bounded functions $f:\N_0\to\R$.
\end{description}
The first identity, \eqref{SteinChen}, is usually referred to as the Stein--Chen identity, in view of the contribution by \citet{chen75}.

\smallskip
Univariate Stein identities, such as those in \eqref{SteinChen}--\eqref{SteinNB}, have found various applications in the literature. They have been used for moment computations \citep{weiss2022}, for developing goodness-of-fit tests \citep{betsch22,weiss2023} as well as control charts \citep{weiss2025}, for generalized types of moment estimation \citep{wang2023,nik2024,ebner25}, and they have also been applied in the time-series case \citep{aleksandrov22,aleksandrov24}. In what follows, however, we are not concerned with univariate count random variables, but we consider \emph{bivariate} count distributions instead. More precisely, we extend the aforementioned univariate Stein identities \eqref{SteinChen}--\eqref{SteinNB} to the respective bivariate case, namely by considering the bivariate Poisson distribution in Section~\ref{Stein Characterization of Bivariate Poisson Distribution}, the bivariate binomial distribution (of type~I) in Section~\ref{Stein Characterization of Bivariate Binomial Distribution}, and the bivariate negative-binomial distribution in Section~\ref{Stein Characterization of Bivariate Negative-binomial Distribution}. Afterwards, Section~\ref{Applications} presents various directions for applying the bivariate Stein characterizations. We demonstrate that our novel Stein identities allow to derive demanding moment expressions, to develop bivariate goodness-of-fit tests as well as tests for the symmetry of a bivariate distribution (i.e., for the exchangeability of the corresponding random variables), and we illustrate their application in practice by a couple of real-world data examples. Finally, we conclude in Section~\ref{Conclusions} and outline opportunities for future research.

\section{Stein Characterization of Bivariate Poisson Distribution}
\label{Stein Characterization of Bivariate Poisson Distribution}
The bivariate Poisson distribution is a probability distribution that models the joint occurrence of two non-negative integer-valued random variables, $(X_1,X_2)$ with range $\N_0^2$, and is often used to describe paired count data with possible cross-correlation. It extends the univariate Poisson distribution by incorporating a covariance structure, making it suitable for applications where two events are related, such as the frequency of defects in two production lines, or the number of goals scored by two teams in a sports match \citep{karlis03}.

\smallskip
Given three independent random variables $Z_0, Z_1,Z_2$ with $Z_i\sim \Poi(\lambda_i)$ for $i=0,1,2$, the bivariate count random vector
\begin{equation}
  \fX = (X_1,X_2)^\top=(Z_0+Z_1,\ Z_0+Z_2)^\top
\end{equation}
is said to be bivariately Poisson (BPoi) distributed according to $\BPoi(\lambda_0; \lambda_1,\lambda_2)$.
The marginals satisfy $X_i\sim\Poi(\lambda_0+\lambda_i)$ for $i=1,2$, and the cross-covariance equals $\operatorname{Cov}(X_1,X_2)=\lambda_0$. The BPoi-distribution is suitable for equidispersed count, i.e., where $\E[X_i]=\V[X_i]$ is satisfied for $i=1,2$.
Characterized by parameters for each variable's mean and their covariance, this distribution provides a flexible framework for analyzing dependent count data in fields like statistics, epidemiology, and engineering, see \citet{kocherlakota14} for further details and references. 

\begin{theorem}[Stein characterization of bivariate Poisson distribution]\label{main-thm}
  Suppose $(X_1,X_2)\sim \BPoi(\lambda_0; \lambda_1,\lambda_2)$, then for any function
  $f: \mathbb{N}_0\times\mathbb{N}_0\to\R$ such that the following expectations exist, the following two identities hold:
  \begin{equation}\label{thm-eq-1}
    \E[X_1\, f(X_1,X_2)]=\lambda_1\, \E[f(X_1+1,X_2)]+\lambda_0\, \E[f(X_1+1,X_2+1)],
  \end{equation}
  \begin{equation}\label{thm-eq-2}
    \E[X_2\, f(X_1,X_2)]=\lambda_2\, \E[f(X_1,X_2+1)]+\lambda_0\, \E[f(X_1+1,X_2+1)].
  \end{equation}
 Conversely, if the identities \eqref{thm-eq-1}--\eqref{thm-eq-2} hold for all $f: \mathbb{N}_0\times\mathbb{N}_0\to\R$, and assuming that $(X_1,X_2)$ is a non-negative integer-valued random vector, then $(X_1,X_2)\sim \BPoi(\lambda_0; \lambda_1,\lambda_2)$.
\end{theorem}
Note that the identities \eqref{thm-eq-1}--\eqref{thm-eq-2} include identities (15)--(16) in \citet{weiss2022} as a special case.

\begin{proof}
  We first suppose $(X_1,X_2)\sim \BPoi(\lambda_0; \lambda_1,\lambda_2)$ and write $p(x,y):=\P(X_1=x,X_2=y)$ as the probability mass function (pmf) of $(X_1,X_2)$. Recall that (see Section~4.1.1 in \citet{kocherlakota14})
\ba
\label{BPoi_prop}
\begin{array}{l}
  x\,p(x,y)=\lambda_1\, p(x-1,y)+\lambda_0\, p(x-1,y-1),
\\[1ex]
  y\,p(x,y)=\lambda_2\, p(x,y-1)+\lambda_0\, p(x-1,y-1),
	\end{array}
\ea
 with the convention $p(x,y)=0$ if $x<0$ or $y<0$. Now, we have
  \begin{align*}
    \E[X_1\,f(X_1,X_2)]&=\sum_{x,y} x\,f(x,y)\,p(x,y)\\
    &=\sum_{x,y} f(x,y)\,[\lambda_1\, p(x-1,y)+\lambda_0\, p(x-1,y-1)]\\
    &=\lambda_1 \sum_{x,y} f(x,y)\,p(x-1,y)+\lambda_0 \sum_{x,y} f(x,y)\,p(x-1,y-1)\\
    &=\lambda_1\, \E[f(X_1+1,X_2)]+\lambda_0\,\E[f(X_1+1,X_2+1)],
  \end{align*}
  which proves \eqref{thm-eq-1}. \eqref{thm-eq-2} can be proved similarly or by a symmetric argument, so we omit the details here.

\smallskip
  Conversely, let \eqref{thm-eq-1}--\eqref{thm-eq-2} hold.
	For an arbitrary function $g: \mathbb{N}_0\to\R$, define $f(x,y)=g(y)$ for all $x\in\N_0$. Then, (\ref{thm-eq-2}) yields
  $$
  \E[X_2\,g(X_2)]=(\lambda_0+\lambda_2)\,\E[g(X_2+1)].
  $$
  Thus, from the Stein--Chen characterization \eqref{SteinChen} of the Poisson distribution, we know $X_2\sim \Poi(\lambda_0+\lambda_2)$; analogously, one proves $X_1\sim \Poi(\lambda_0+\lambda_1)$.
	
	Next, we take $f(x,y)=s^xt^y$, $|s|,|t|<1$, and denote the joint probability generating function (pgf) of $(X_1,X_2)$ as $G(s,t)=\E[s^{X_1}t^{X_2}], |s|,|t|<1$. Then, from (\ref{thm-eq-1}), we obtain
  \begin{align*}
    s\,\frac{\partial}{\partial s}G(s,t)=\E[X_1\,s^{X_1}t^{X_2}]=\lambda_1\, s\, G(s,t)+\lambda_0\, st\, G(s,t),
  \end{align*}
  that is
  \begin{align*}
    \frac{\partial}{\partial s}G(s,t)=\lambda_1\, G(s,t)+\lambda_0\, t\, G(s,t)= (\lambda_1+\lambda_0\, t)\, G(s,t)
  \end{align*}
  holds for all $|s|,|t|<1$. The above partial differential equation can be solved by
 \begin{align*}
    \frac{\partial}{\partial s}\ln G(s,t)=\lambda_1 +\lambda_0 t,
  \end{align*}
  which further implies $G(s,t)=\exp\{\lambda_1s +\lambda_0 st+\theta(t)\}$, where $\theta(t)$ is a function that depends only on $t$. Observe that on the one hand, $G(1,t)=\E[t^{X_2}]=\exp\big((\lambda_0+\lambda_2)(t-1)\big)$ because of $X_2\sim \Poi(\lambda_0+\lambda_2)$. On the other hand, $G(1,t)=\exp\{\lambda_1 +\lambda_0 t+\theta(t)\}$. Thus, we get $\theta(t)=(t-1)\lambda_2-\lambda_0-\lambda_1$. So finally,
  $$
  G(s,t)=\exp\Big\{\lambda_1(s-1) +\lambda_2(t-1)+\lambda_0 (st-1)\Big\},\quad |s|,|t|<1,
  $$
  which uniquely determines the joint pmf of $(X_1,X_2)\sim \BPoi(\lambda_0; \lambda_1,\lambda_2)$. This completes the proof of Theorem~\ref{main-thm}.
\end{proof}

\begin{rmk}
  In the boundary case $\lambda_0+\lambda_2=0$, i.e., if $X_2\equiv 0$, then \eqref{thm-eq-1} reduces to the well-known Stein characterization \eqref{SteinChen} of the univariate Poisson distribution.
\end{rmk}

\section{Stein Characterization of Bivariate Binomial Distribution}
\label{Stein Characterization of Bivariate Binomial Distribution}
Let $\fZ_k:=(Z_{k1},Z_{k2})^\top$, $k\geq 1$, be independent and identically distributed (i.i.d.) bivariate Bernoulli random vectors, with the common joint pmf of $\fZ_k$ being denoted by
\[
p_{ij} := \P(\fZ_k = (i,j)^\top), \quad i,j \in \{0,1\}.
\]
The parameters $a_1, a_2 \in (0,1)$ and $0<a <\min\{a_1, a_2\}$ are chosen such that
\begin{equation}\label{eq-joint-pmf}
p_{11}=a, \quad p_{10}+p_{11}=a_1, \quad p_{01}+p_{11}=a_2
\end{equation}
defines a valid pmf. In particular, $1-a_1=P(Z_{k1}=0)\geq p_{01}=a_2-a$ has to hold. Equivalently, one may use the correlation parameter
\[
\phi_a := \operatorname{Corr}(Z_{k1}, Z_{k2}) =
\frac{a - a_1 a_2}{\sqrt{a_1 a_2 (1-a_1)(1-a_2)}}.
\]
Define the random vector
\[
\fX = (X_1, X_2)^\top := \sum_{k=1}^N \fZ_k
\quad\text{with } N\in\N.
\]
Then, $\fX$ is said to follow the \emph{bivariate binomial (BVB) distribution of type~I} with parameters $N, a_1, a_2, \phi_a$, denoted by
\[
\fX \sim \BVB(N; a_1, a_2, \phi_a).
\]
Its range is $\{0,1,\dots,N\}^2$, the marginal distributions are $X_i \sim \Bin(N, a_i)$ for $i=1,2$, and $\operatorname{Corr}(X_1,X_2) = \phi_a$.
According to \citet{kocherlakota14}, the joint pgf of $\fX$ is
\begin{equation}\label{def-mgf}
G(z_1, z_2) := \E\bigl[z_1^{X_1} z_2^{X_2}\bigr]
    = \bigl( p_{00} + p_{10} z_1 + p_{01} z_2 + p_{11} z_1 z_2 \bigr)^N,
\end{equation}
where $p_{00}, p_{10}, p_{01}, p_{11}$ are determined by \eqref{eq-joint-pmf}.

\begin{theorem}[Stein characterization of bivariate binomial distribution]\label{thm:stein-bvb}
Let $N\in\mathbb{N}$ be fixed and \( (X_1, X_2) \) be a non-negative random vector with range $\{0,1,\ldots,N\}^2$. Then, \( (X_1, X_2) \sim \BVB(N; a_1, a_2, \phi_a) \) if and only if for all functions \( f: \mathbb{N}_0\times\mathbb{N}_0 \to \R \), the following two identities hold:
\begin{align}\label{eq:stein1}
p_{00}\,& \E[X_1 f(X_1, X_2)] + p_{01}\, \E[X_1 f(X_1, X_2+1)] \\
&= p_{10}\, \E[(N-X_1)\, f(X_1+1, X_2)] + p_{11}\, \E[(N-X_1)\, f(X_1+1, X_2+1)], \notag
\end{align}
and
\begin{align}\label{eq:stein2}
p_{00}\,& \E[X_2 f(X_1, X_2)] + p_{10}\, \E[X_2 f(X_1+1, X_2)] \\
&= p_{01}\, \E[(N-X_2)\, f(X_1, X_2+1)] + p_{11}\, \E[(N-X_2)\, f(X_1+1, X_2+1)],\notag
 \end{align}
where the parameters $a_1$, $a_2$, and $\phi_a$ are uniquely determined by $p_{ij}$ with $i,j\in\{0,1\}$ via \eqref{eq-joint-pmf}.
\end{theorem}

\begin{proof}
We first prove the necessity part. Let $p(x_1, x_2)$ denote the joint pmf of $\BVB(N; a_1, a_2, \phi_a)$, and let $G(z_1, z_2)$ be its joint pgf defined in \eqref{def-mgf}. Differentiating $G$ with respect to $z_1$ and then multiplying both sides by $P(z_1, z_2)$ yields
\begin{equation}\label{eq-1}
P(z_1, z_2) \frac{\partial G}{\partial z_1} = N (p_{10} + p_{11} z_2)\, G(z_1, z_2),
\end{equation}
where $P(z_1, z_2) = p_{00} + p_{10} z_1 + p_{01} z_2 + p_{11} z_1 z_2$ is the pgf of the bivariate Bernoulli distribution \eqref{eq-joint-pmf}. We rewrite the left-hand side of \eqref{eq-1} as
\begin{align}\label{eq-2}
P(z_1, z_2) \frac{\partial G}{\partial z_1}
&= (p_{00} + p_{10} z_1 + p_{01} z_2 + p_{11} z_1 z_2)
   \sum_{x_1, x_2} x_1 p(x_1, x_2)\, z_1^{x_1-1} z_2^{x_2} \notag\\
&= \sum_{x_1, x_2} x_1 p(x_1, x_2) \bigl[ p_{00}\, z_1^{x_1-1} z_2^{x_2}
   + p_{10}\, z_1^{x_1} z_2^{x_2}
   + p_{01}\, z_1^{x_1-1} z_2^{x_2+1}
   + p_{11}\, z_1^{x_1} z_2^{x_2+1} \bigr].
\end{align}
Then, we reindex each term as follows:
\begin{itemize}
    \item $p_{00} \sum_{x_1, x_2} x_1 p(x_1, x_2) z_1^{x_1-1} z_2^{x_2}
          = p_{00} \sum_{x_1, x_2} (x_1+1) p(x_1+1, x_2) z_1^{x_1} z_2^{x_2}$,
    \item $p_{10} \sum_{x_1, x_2} x_1 p(x_1, x_2) z_1^{x_1} z_2^{x_2}
          = p_{10} \sum_{x_1, x_2} x_1 p(x_1, x_2) z_1^{x_1} z_2^{x_2}$,
    \item $p_{01} \sum_{x_1, x_2} x_1 p(x_1, x_2) z_1^{x_1-1} z_2^{x_2+1}
          = p_{01} \sum_{x_1, x_2} (x_1+1) p(x_1+1, x_2-1) z_1^{x_1} z_2^{x_2}$,
    \item $p_{11} \sum_{x_1, x_2} x_1 p(x_1, x_2) z_1^{x_1} z_2^{x_2+1}
          = p_{11} \sum_{x_1, x_2} x_1 p(x_1, x_2-1) z_1^{x_1} z_2^{x_2}$.
\end{itemize}
The right-hand side of \eqref{eq-1} is given by
\begin{align}\label{eq-3}
N(p_{10} + p_{11} z_2) G(z_1, z_2)
&= N p_{10} \sum_{x_1, x_2} p(x_1, x_2) z_1^{x_1} z_2^{x_2}
+ N p_{11} \sum_{x_1, x_2} p(x_1, x_2) z_1^{x_1} z_2^{x_2+1}\notag \\
&= \sum_{x_1, x_2} \bigl[ N p_{10} p(x_1, x_2)
+ N p_{11} p(x_1, x_2-1) \bigr] z_1^{x_1} z_2^{x_2}.
\end{align}
Equating coefficients of $z_1^{x_1} z_2^{x_2}$ in \eqref{eq-2} and \eqref{eq-3}, we obtain the recurrence relation
\begin{align*}
 p_{00} (x_1+1) p(x_1+1, x_2) + p_{10} x_1 p(x_1, x_2)
& + p_{01} (x_1+1) p(x_1+1, x_2-1) + p_{11} x_1 p(x_1, x_2-1) \\
&= N p_{10} p(x_1, x_2) + N p_{11} p(x_1, x_2-1).
\end{align*}
Multiplying both sides of the above relation by an arbitrary test function $f(x_1+1, x_2)$, summing over all $x_1, x_2$ and re-indexing each term, we get that the left-hand side is equal to the sum of the following four terms:
\begin{itemize}
    \item Term 1: Let $y_1 = x_1+1$, then
    \[
    p_{00} \sum_{y_1, x_2} f(y_1, x_2) y_1 p(y_1, x_2) = p_{00} \E[X_1 f(X_1, X_2)].
    \]
    \item Term 2: We have
    \[
    p_{10} \sum_{x_1, x_2} f(x_1+1, x_2) x_1 p(x_1, x_2) = p_{10} \E[X_1 f(X_1+1, X_2)].
    \]
    \item Term 3: Let $y_1 = x_1+1$, $y_2 = x_2-1$, then
    \[
    p_{01} \sum_{y_1, y_2} f(y_1, y_2+1) y_1 p(y_1, y_2) = p_{01} \E[X_1 f(X_1, X_2+1)].
    \]
    \item Term 4: Let $y_2 = x_2-1$, then
    \[
    p_{11} \sum_{x_1, y_2} f(x_1+1, y_2+1) x_1 p(x_1, y_2) = p_{11} \E[X_1 f(X_1+1, X_2+1)].
    \]
\end{itemize}
The corresponding right-hand side equals
\begin{align*}
\sum_{x_1, x_2}f(x_1+1, x_2) &\bigl[ N p_{10} p(x_1, x_2) + N p_{11} p(x_1, x_2-1) \bigr] \\
&= N p_{10} \sum_{x_1, x_2} f(x_1+1, x_2) p(x_1, x_2)
+ N p_{11} \sum_{x_1, x_2} f(x_1+1, x_2) p(x_1, x_2-1) \\
&= N p_{10} \E[f(X_1+1, X_2)] + N p_{11} \E[f(X_1+1, X_2+1)].
\end{align*}
Thus, we finally get
\begin{align}
p_{00} \E[X_1 f(X_1, X_2)]& + p_{10} \E[X_1 f(X_1+1, X_2)]\notag \\
&+ p_{01} \E[X_1 f(X_1, X_2+1)] + p_{11} \E[X_1 f(X_1+1, X_2+1)]\notag \\
&= N p_{10} \E[f(X_1+1, X_2)] + N p_{11} \E[f(X_1+1, X_2+1)], \notag
\end{align}
so \eqref{eq:stein1} follows. \eqref{eq:stein2} is proved similarly or by a symmetric argument.

Next, we prove the sufficiency part. To determine the joint distribution, consider the joint pgf
\( G(s,t) = \E[ s^{X_1} t^{X_2} ] \) with \( |s|,|t| \le 1 \).
First, we choose $f(x,y) = s^{x} t^{y}$ in (\ref{eq:stein1}) and compute each term:
\begin{align*}
\E[X_1 f(X_1, X_2)] &= \E[X_1 s^{X_1} t^{X_2}] = s\frac{\partial G}{\partial s}(s,t), \\
\E[X_1 f(X_1+1, X_2)] &= \E[X_1 s^{X_1+1} t^{X_2}] = s^2\frac{\partial G}{\partial s}(s,t), \\
\E[X_1 f(X_1, X_2+1)] &= \E[X_1 s^{X_1} t^{X_2+1}] = st\frac{\partial G}{\partial s}(s,t), \\
\E[X_1 f(X_1+1, X_2+1)] &= \E[X_1 s^{X_1+1} t^{X_2+1}] = s^2t\frac{\partial G}{\partial s}(s,t), \\
\E[f(X_1+1, X_2)] &= \E[s^{X_1+1} t^{X_2}] = s\,G(s,t), \\
\E[f(X_1+1, X_2+1)] &= \E[s^{X_1+1} t^{X_2+1}] = st\,G(s,t).
\end{align*}
Substituting into (\ref{eq:stein1}):
\[
p_{00} s\frac{\partial G}{\partial s} + p_{10} s^2\frac{\partial G}{\partial s} + p_{01} st\frac{\partial G}{\partial s} + p_{11} s^2t\frac{\partial G}{\partial s}
= N p_{10} s\, G + N p_{11} st\, G.
\]
For $s \neq 0$, we cancel out $s$:
\begin{equation}\label{A}
\left[p_{00} + p_{10}s + p_{01}t + p_{11}st\right]\frac{\partial G}{\partial s} = N(p_{10} + p_{11}t)G.
\end{equation}
Similarly, from \eqref{eq:stein2}, we obtain
\begin{equation}\label{B}
\left[p_{00} + p_{10}s + p_{01}t + p_{11}st\right]\frac{\partial G}{\partial t} = N(p_{01} + p_{11}s)G.
\end{equation}
Denoting the Bernoulli pgf by $P(s,t) = p_{00} + p_{10}s + p_{01}t + p_{11}st$ again, the equations \eqref{A} and \eqref{B} imply
\begin{align}\label{eq-C1}
\frac{\partial \log G}{\partial s} = \frac{N(p_{10} + p_{11}t)}{P(s,t)}, \quad
\frac{\partial \log G}{\partial t} = \frac{N(p_{01} + p_{11}s)}{P(s,t)}.
\end{align}
Integrating the first equation with respect to $s$ (treating $t$ as a constant), we obtain
\[
\log G(s,t) = N\log P(s,t) + H(t),
\]
where $H(t)$ is a function of $t$ only. So differentiating this expression with respect to $t$, and
comparing it to the second expression in \eqref{eq-C1}, we get $H'(t) = 0$ for all $t$. Thus, $H(t)$ is a constant, say $H(t) = C$. Since bivariate pgfs satisfy $G(1,1) = 1$ and $P(1,1) = 1$, we must have $H(1) = 0$. Therefore, $H(t)\equiv 0$ and, as a consequence,
\[
G(s,t) = [p_{00} + p_{10}s + p_{01}t + p_{11}st]^N
\]
holds for $|s|\leq 1,|t|\leq 1$. This is exactly the pgf of $\BVB(N; a_1, a_2, \phi_a)$, where $a_1 = p_{10} + p_{11}$, $a_2 = p_{01} + p_{11}$, and
$$
\phi_a=
\frac{p_{11} - a_1 a_2}{\sqrt{a_1 a_2 (1-a_1)(1-a_2)}}.
$$
Since the pgf uniquely determines the distribution, we conclude $(X_1, X_2) \sim \BVB(N; a_1, a_2, \phi_a)$, and the proof of Theorem~\ref{thm:stein-bvb} is complete.
\end{proof}

\begin{rmk}
In the boundary case $a_2=p_{11}+p_{01}=0$, i.e., if $X_2\equiv 0$, the Stein characterization in \eqref{eq:stein1} reduces to
$$
(1-a_1)\E[Xf_0(X)]=a_1\E[(N-X)f_0(X+1)]
$$
for any bounded function~$f_0$ defined on $\N_0$. This is exactly the Stein characterization of $X\sim \Bin(N,a_1)$, recall \eqref{SteinBin}.
\end{rmk}

\begin{rmk}\label{rmk-2Pois-1Pois}
  Take \( f(x,y) = g(y) \) independent of \( x \) in (\ref{eq:stein2}). Then,
\begin{align*}
(p_{00}+p_{10})\,\E[ X_2\, g(X_2) ] &=
(p_{01}+p_{11})\,\E[(N-X_2)\,g(X_2+1)],
\end{align*}
that is,
$$
(1-a_2)\,\E[ X_2\, g(X_2) ]=a_2\,\E[(N-X_2)\, g(X_2+1) ]
$$
where we used $a_2 = p_{01} + p_{11}$. This is exactly the classical Stein identity for $X_2 \sim \Bin(N, a_2)$, recall \eqref{SteinBin}, so $X_2 \sim \Bin(N, a_2)$. Similarly, taking $f(x,y) = h(x)$ in (\ref{eq:stein1}) gives $X_1 \sim \Bin(N, a_1)$.
\end{rmk}

\begin{rmk}
Suppose $Np_{11}\to \lambda_0$, $Np_{10}\to \lambda_1$ and $Np_{01}\to \lambda_2$ for $N\to\infty$, which also implies that $p_{10},p_{01},p_{11}\to 0$ and $p_{00}\to 1$. Then, by the Poisson approximation, we have $(X_1,X_2)\to_d \BPoi(\lambda_0; \lambda_1,\lambda_2)$, and \eqref{eq:stein1}--\eqref{eq:stein2} reduce
to \eqref{thm-eq-1}--\eqref{thm-eq-2}. So the result in Theorem~\ref{thm:stein-bvb} is compatible with Theorem~\ref{main-thm}.
\end{rmk}

\section{Stein Characterization of Bivariate Negative-binomial Distribution}
\label{Stein Characterization of Bivariate Negative-binomial Distribution}
With \(\nu \in (0,\infty)\), \(\pi_{1},\pi_{2}\in (0,1)\), and \(\pi_{0}\in (- \pi_{1}\pi_{2},1)\) such that \(\pi_{\bullet}:= \sum_{i = 0}^{2}\pi_{i}< 1\) holds, the bivariate negative-binomial (BNB) distribution \(\BNB(\nu;\pi_{1},\pi_{2},\pi_{0})\) is defined by the bivariate pgf
\begin{equation}\label{BNB-pgf}
G(z_{1},z_{2}) = \left(\frac{1}{1 - \pi_{\bullet}} -\frac{\pi_{1}}{1 - \pi_{\bullet}} z_{1} - \frac{\pi_{2}}{1 - \pi_{\bullet}} z_{2} - \frac{\pi_{0}}{1 - \pi_{\bullet}} z_{1}z_{2}\right)^{-\nu},
\end{equation}
see Section~5 in \citet{kocherlakota14} or Section~2.2 in \citet{aleksandrov24} for details. The special case $\pi_0=0$ is also known as the negative-multinomial distribution.
We have a Poisson limit for \(\nu\rightarrow \infty\), namely \(\BPoi(\lambda_{0};\lambda_{1},\lambda_{2})\) if \(\nu\pi_i/(1 - \pi_{\bullet})\rightarrow \lambda_{i}\) for \(\nu\rightarrow \infty\). The marginals are univariate NB-distributions, namely \(X_{1}\sim \mathrm{NB}(\nu, \frac{1 - \pi_{\bullet}}{1 - \pi_{2}})\) and \(X_{2}\sim \mathrm{NB}(\nu, \frac{1 - \pi_{\bullet}}{1 - \pi_{1}})\), respectively. The corresponding marginal means are
$\mu_1 = \nu\, \frac{\pi_{0}+\pi_1}{1 - \pi_{\bullet}}$
and
$\mu_2 = \nu\, \frac{\pi_{0}+\pi_2}{1 - \pi_{\bullet}}$, respectively, and the cross-covariance equals $\operatorname{Cov}(X_1, X_2)= \nu\, \frac{\pi_0+\pi_1 \pi_2}{(1-\pi_\bullet)^2}$. For simulating the \(\BNB(\nu;\pi_{1},\pi_{2},\pi_{0})\)-distribution, it is useful to know that the distribution of~$X_1|x_2$ is a convolution of the distributions $\Bin(x_2,\frac{\pi_{0}}{\pi_2+\pi_{0}})$ and $\NB(\nu+x_2, 1-\pi_1)$ if $\pi_0>0$, and that $X_1|x_2 \sim \NB(\nu+x_2, 1-\pi_1)$ if $\pi_0=0$.

\begin{theorem}[Stein characterization of bivariate NB distribution]\label{thm:stein-bnb} Let \(\nu\in (0,\infty)\) be fixed and \((X_{1},X_{2})\) be a non-negative random vector valued in \(\N_0^{2}\) . Then \((X_{1},X_{2})\sim \BNB(\nu;\pi_{1},\pi_{2},\pi_{0})\) if and only if for all functions \(f:\N_0^{2}\to \mathbb{R}\), provided that the following moments exist, the following two identities hold:
\begin{align}
\mathbb{E}[X_1\, f(X_1, X_2)] &- \pi_2\, \mathbb{E}[X_1\, f(X_1, X_2+1)] \label{eq:BNB-Stein1}\\
 &= \pi_1\, \mathbb{E}[(\nu+X_1)\,f(X_1+1, X_2)] + \pi_0\, \mathbb{E}[(\nu+X_1)\,f(X_1+1, X_2+1)] \notag
\end{align}
and
\begin{align}
\mathbb{E}[X\,_2 f(X_1, X_2)] &- \pi_1\, \mathbb{E}[X_2\, f(X_1+1, X_2)] \label{eq:BNB-Stein2}\\
 &= \pi_2\, \mathbb{E}[(\nu+X_2)\,f(X_1, X_2+1)] + \pi_0\, \mathbb{E}[(\nu+X_2)\,f(X_1+1, X_2+1)].\notag
\end{align}
\end{theorem}

\begin{proof}
We first prove the necessity part. Suppose \((X_1, X_2) \sim \BNB(\nu; \pi_1, \pi_2, \pi_0)\). Let \(p(x_1, x_2) = \mathbb{P}(X_1 = x_1, X_2 = x_2)\) denote the joint pmf. From \eqref{BNB-pgf}, we have
\begin{equation}
G(s,t) = \left(\frac{1 - \pi_\bullet}{1 - \pi_1 s - \pi_2 t - \pi_0 s t}\right)^\nu = \sum_{x_1, x_2} p(x_1, x_2) s^{x_1} t^{x_2}. \label{eq:G-series}
\end{equation}

Differentiating with respect to \(s\) gives:
\begin{equation}
(1 - \pi_1 s - \pi_2 t - \pi_0 s t) \frac{\partial G}{\partial s} = \nu(\pi_1 + \pi_0 t)\, G(s,t). \label{eq:pde-s}
\end{equation}

Now, substituting the series expansion from \eqref{eq:G-series} into \eqref{eq:pde-s} yields
\begin{align*}
\sum_{x_1, x_2} & x_1 p(x_1, x_2) s^{x_1-1} t^{x_2}-\pi_1 \sum_{x_1, x_2} x_1 p(x_1, x_2) s^{x_1} t^{x_2} -\pi_2 \sum_{x_1, x_2} x_1 p(x_1, x_2) s^{x_1-1} t^{x_2+1}\\
 &\qquad-\pi_0 \sum_{x_1, x_2} x_1 p(x_1, x_2) s^{x_1} t^{x_2+1} \\
&= \nu\pi_1 \sum_{x_1, x_2} p(x_1, x_2) s^{x_1} t^{x_2} + \nu\pi_0 \sum_{x_1, x_2} p(x_1, x_2) s^{x_1} t^{x_2+1}.
\end{align*}
Multiplying both sides by \(s\) and comparing the coefficients of \(s^{x_1} t^{x_2}\), we get the following relation:
\begin{align}\label{eq:recurrence1}
x_1 p(x_1, x_2) &- \pi_1 (x_1-1) p(x_1-1, x_2) - \pi_2 x_1 p(x_1, x_2-1) - \pi_0 (x_1-1) p(x_1-1, x_2-1) \notag\\
&= \nu\pi_1 p(x_1-1, x_2) + \nu\pi_0 p(x_1-1, x_2-1).
\end{align}
Next, we use \eqref{eq:recurrence1} for expressing expectations. Multiplying \eqref{eq:recurrence1} by \(f(x_1, x_2)\), summing over all possible \(x_1, x_2\), and re-indexing each sum, we get
\begin{align*}
&\sum_{x_1, x_2} x_1 p(x_1, x_2) f(x_1, x_2) = \mathbb{E}[X_1 f(X_1, X_2)], \\
&\sum_{x_1, x_2} (x_1-1) p(x_1-1, x_2) f(x_1, x_2) = \sum_{y_1, x_2} y_1 p(y_1, x_2) f(y_1+1, x_2) = \mathbb{E}[X_1 f(X_1+1, X_2)], \\
&\sum_{x_1, x_2} x_1 p(x_1, x_2-1) f(x_1, x_2) = \sum_{x_1, y_2} x_1 p(x_1, y_2) f(x_1, y_2+1) = \mathbb{E}[X_1 f(X_1, X_2+1)], \\
&\sum_{x_1, x_2} (x_1-1) p(x_1-1, x_2-1) f(x_1, x_2) = \sum_{y_1, y_2} y_1 p(y_1, y_2) f(y_1+1, y_2+1) = \mathbb{E}[X_1 f(X_1+1, X_2+1)], \\
&\sum_{x_1, x_2} p(x_1-1, x_2) f(x_1, x_2) = \sum_{y_1, x_2} p(y_1, x_2) f(y_1+1, x_2) = \mathbb{E}[f(X_1+1, X_2)], \\
&\sum_{x_1, x_2} p(x_1-1, x_2-1) f(x_1, x_2) = \sum_{y_1, y_2} p(y_1, y_2) f(y_1+1, y_2+1) = \mathbb{E}[f(X_1+1, X_2+1)].
\end{align*}
As a consequence, we have
\begin{align*}
\mathbb{E}[X_1 f(X_1, X_2)] &- \pi_1 \mathbb{E}[X_1 f(X_1+1, X_2)] \\
&- \pi_2 \mathbb{E}[X_1 f(X_1, X_2+1)] - \pi_0 \mathbb{E}[X_1 f(X_1+1, X_2+1)]\notag\\
 &= \nu\pi_1 \mathbb{E}[f(X_1+1, X_2)] + \nu\pi_0 \mathbb{E}[f(X_1+1, X_2+1)],
\end{align*}
which proves \eqref{eq:BNB-Stein1}. By a symmetric argument, \eqref{eq:BNB-Stein2} is obtained similarly.

For the sufficiency part, assume \eqref{eq:BNB-Stein1} and \eqref{eq:BNB-Stein2} hold for all bounded \(f\). Take \(f(x,y) = s^x t^y\) with \(|s|, |t| < 1\), and let \(G(s,t) = \mathbb{E}[s^{X_1} t^{X_2}]\) be the corresponding joint pgf of $(X_1,X_2)$. Then,
\begin{align*}
\mathbb{E}[X_1 s^{X_1} t^{X_2}] &= s\frac{\partial G}{\partial s}, \\
\mathbb{E}[X_1 s^{X_1+1} t^{X_2}] &= s^2\frac{\partial G}{\partial s}, \\
\mathbb{E}[X_1 s^{X_1} t^{X_2+1}] &= st\frac{\partial G}{\partial s}, \\
\mathbb{E}[X_1 s^{X_1+1} t^{X_2+1}] &= s^2 t\frac{\partial G}{\partial s}, \\
\mathbb{E}[s^{X_1+1} t^{X_2}] &= s\, G(s,t), \\
\mathbb{E}[s^{X_1+1} t^{X_2+1}] &= st\, G(s,t).
\end{align*}
Substituting these into \eqref{eq:BNB-Stein1}, we get following partial differential equation:
\begin{equation}
s\frac{\partial G}{\partial s} - \pi_1 s^2\frac{\partial G}{\partial s} - \pi_2 st\frac{\partial G}{\partial s} - \pi_0 s^2 t\frac{\partial G}{\partial s} = \nu\pi_1 s\, G + \nu\pi_0 st\, G.
\end{equation}
Division by \(s\), for \(s \neq 0\), yields
\begin{equation}
(1 - \pi_1 s - \pi_2 t - \pi_0 s t)\frac{\partial G}{\partial s} = \nu(\pi_1 + \pi_0 t)\, G. \label{eq:pde-s-2}
\end{equation}
Similarly, from \eqref{eq:BNB-Stein2}, we obtain
\begin{equation}
(1 - \pi_1 s - \pi_2 t - \pi_0 s t)\frac{\partial G}{\partial t} = \nu(\pi_2 + \pi_0 s)\, G. \label{eq:pde-t-2}
\end{equation}
Dividing \eqref{eq:pde-s-2} and \eqref{eq:pde-t-2} by \(G\) yields
\begin{align}\label{eq:log-s}
\frac{\partial \log G}{\partial s} = \frac{\nu(\pi_1 + \pi_0 t)}{1 - \pi_1 s - \pi_2 t - \pi_0 s t},\quad
\frac{\partial \log G}{\partial t} = \frac{\nu(\pi_2 + \pi_0 s)}{1 - \pi_1 s - \pi_2 t - \pi_0 s t}.
\end{align}
Integrating the first equation in \eqref{eq:log-s} with respect to \(s\) (treating \(t\) as constant), we get
\begin{equation}
\log G(s,t) = -\nu \log(1 - \pi_1 s - \pi_2 t - \pi_0 s t) + H(t), \label{eq:logG}
\end{equation}
where \(H(t)\) depends only on \(t\). Differentiating \eqref{eq:logG} with respect to \(t\) and comparing with the second equation of \eqref{eq:log-s} gives
\begin{equation}
\frac{\nu(\pi_2 + \pi_0 s)}{1 - \pi_1 s - \pi_2 t - \pi_0 s t} + H'(t) = \frac{\nu(\pi_2 + \pi_0 s)}{1 - \pi_1 s - \pi_2 t - \pi_0 s t}.
\end{equation}
Thus, \(H'(t) = 0\), so \(H(t)\) is constant. Since \(G(1,1) = 1\), \eqref{eq:logG} implies the equality
\begin{equation}
0 = -\nu \log(1 - \pi_1 - \pi_2 - \pi_0) + H(1).
\end{equation}
Therefore \(H(t) \equiv \nu \log(1 - \pi_\bullet)\), and thus
\begin{equation}
G(s,t) = \left(\frac{1 - \pi_\bullet}{1 - \pi_1 s - \pi_2 t - \pi_0 s t}\right)^\nu,
\end{equation}
which is exactly \eqref{BNB-pgf}. This completes the proof of Theorem~\ref{thm:stein-bnb}.
\end{proof}

\begin{rmk}
When \(\nu \to \infty\) and \(\nu\pi_i \to \lambda_i\) for \(i=0,1,2\) (hence, $\pi_i\to 0$ for \(i=0,1,2\)), then \eqref{eq:BNB-Stein1} and \eqref{eq:BNB-Stein2} reduce to
\begin{align*}
\mathbb{E}[X_1 f(X_1, X_2)] &= \lambda_1 \mathbb{E}[f(X_1 + 1, X_2)] + \lambda_0 \mathbb{E}[f(X_1 + 1, X_2 + 1)], \\
\mathbb{E}[X_2 f(X_1, X_2)] &= \lambda_2 \mathbb{E}[f(X_1, X_2 + 1)] + \lambda_0 \mathbb{E}[f(X_1 + 1, X_2 + 1)],
\end{align*}
which are exactly the Stein characterizations for the bivariate Poisson distribution in Theorem \ref{main-thm}.
\end{rmk}

\begin{rmk}
In the boundary case $\pi_2 = \pi_0 = 0$, the $\BNB(\nu;\pi_{1},\pi_{2},\pi_{0})$ distribution reduces to a univariate NB-distribution for $X_1$ and $X_2 \equiv 0$. In this case, $\pi_\bullet = \pi_1$, and the Stein characterization \eqref{eq:BNB-Stein1} becomes
\begin{equation}
\mathbb{E}[X_1 g(X_1)] =\pi_1 \mathbb{E}[(\nu+X_1) g(X_1+1)]
\end{equation}
for some univariate bounded function $g$. This is the classical Stein characterization for $\NB(\nu,1-\pi_1)$, recall \eqref{SteinNB}.  We can obtain a similar result if we take functions $f(x,y)$ that depend on only one of the variables~$x,y$, in analogy to Remark \ref{rmk-2Pois-1Pois}.
\end{rmk}

\section{Applications}
\label{Applications}
As briefly surveyed in Section~\ref{Introduction}, Stein characterizations of distributions generally offer a broad potential for applications. In what follows, we briefly sketch possible application areas of the Stein characterizations for the bivariate Poisson, binomial, and negative-binomial distribution as provided by Theorems~\ref{main-thm}--\ref{thm:stein-bnb}, whereas a more comprehensive investigation will be conducted in separate articles as part of future research.

\subsection{Moment Calculations}
\label{Moment Calculations}
In \citet{weiss2022}, Stein identities of the bivariate Poisson distribution were used for moment computations, namely by deriving recursive schemes for certain families of joint moments. Since their investigations have been quite comprehensive, we just add one illustrative example to the $\BPoi(\lambda_0; \lambda_1,\lambda_2)$-derivations of \citet{weiss2022}, namely concerning the quadratic deviation between~$X_1$ and~$X_2$. Taking the difference of equations \eqref{thm-eq-1} and \eqref{thm-eq-2}, we obtain
  \begin{equation}\label{thm-eq-3}
    \E[(X_1-X_2)\, f(X_1,X_2)]=\lambda_1\, \E[f(X_1+1,X_2)] - \lambda_2\, \E[f(X_1,X_2+1)].
  \end{equation}
Hence, using $f(x,y)=x-y$, it immediately follows from \eqref{thm-eq-3} that
  \begin{equation}\label{BPoi_Euclidean}
    \E[(X_1-X_2)^2]
		=\lambda_1\, \E[X_1+1-X_2] - \lambda_2\, \E[X_1-X_2-1]
		=\lambda_1 + \lambda_2 + (\lambda_1 - \lambda_2)^2.
  \end{equation}
The expected absolute difference, in turn, follows by using the sign function $f(x,y)=\textrm{sgn}(x-y)$ instead. 

Considerably fewer moment formulae have been derived for the bivariate binomial and NB distribution so far.
The novel Stein identities for the bivariate binomial distribution in Section~\ref{Stein Characterization of Bivariate Binomial Distribution} and for the bivariate NB distribution in Section~\ref{Stein Characterization of Bivariate Negative-binomial Distribution} can be applied in a similar way as before, which is illustrated in the sequel for some examples.

First, defining the function $f(x_1,x_2)=\indfct(x_1=x)\cdot\indfct(x_2=y)$ for fixed values $x,y\in\N$, where $\indfct(\cdot)$ denotes the indicator function, the Stein characterization of $\BVB(N; a_1, a_2, \phi_a)$ according to Theorem~\ref{thm:stein-bvb} implies the following recursive schemes for probability calculations:
\ba
\label{BVB_prob}
\begin{array}{r@{\,}l}
p_{00}\, x\,p(x,y) =& -p_{01}\, x\,p(x,y-1) \\
&+ p_{10}\, (N-x+1) p(x-1,y) + p_{11}\, (N-x+1)\, p(x-1,y-1), \\[1ex]
p_{00}\, y\,p(x,y) =& -p_{10}\, y\,p(x-1,y) \\
&+ p_{01}\, (N-y+1)\,p(x,y-1) + p_{11}\, (N-y+1)\,p(x-1,y-1),
\end{array}
\ea
which is analogous to the recursive scheme \eqref{BPoi_prop} for the bivariate Poisson distribution.
Similarly, Theorem~\ref{thm:stein-bnb} for $\BNB(\nu;\pi_{1},\pi_{2},\pi_{0})$ implies that
\ba
\label{BNB_prob}
\begin{array}{r@{\,}l}
x\,p(x,y) =& \pi_2\, x\,p(x,y-1) \\
 &+ \pi_1\, (\nu+x-1)\,p(x-1,y) + \pi_0\, (\nu+x-1)\,p(x-1, y-1), \\[1ex]
y\, p(x,y) =& \pi_1\, y\, p(x-1,y)  \\
 &+ \pi_2\, (\nu+y-1)\,p(x, y-1)] + \pi_0\, (\nu+y-1)\,p(x-1, y-1).
\end{array}
\ea
Both schemes \eqref{BVB_prob} and \eqref{BNB_prob} are useful for probability calculations in practice, because the closed formulae for the respective pmfs of $\BVB(N; a_1, a_2, \phi_a)$ and $\BNB(\nu;\pi_{1},\pi_{2},\pi_{0})$ as given in \citet{kocherlakota14} are rather sophisticated and computationally demanding.

Next, let us consider joint factorial moments of $\BVB(N; a_1, a_2, \phi_a)$ and $\BNB(\nu;\pi_{1},\pi_{2},\pi_{0})$. Denote the falling factorials by $x_{(k)}=x\cdots (x-k+1)$ with $x_{(0)}=1$ and the corresponding bivariate factorial moments by $\mu_{(r,s)}:=\E[(X_1)_{(r)}\,(X_2)_{(s)}]$. Obviously, the marginal factorial moments arise by setting either $s=0$ or $r=0$.
Now, inserting $f(x,y)=(x-1)_{(r-1)}\, y_{(s)}$ into \eqref{eq:stein1}, we obtain
\begin{align}
p_{00}\,\mathbb{E}[(X_1)_{(r)}\,&\, (X_2)_{(s)}] + p_{01}\, \mathbb{E}[(X_1)_{(r)}\, (X_2+1)_{(s)}] \notag \\
 &= p_{10}\, \mathbb{E}[(N-X_1)\,(X_1)_{(r-1)}\, (X_2)_{(s)}] + p_{11}\, \mathbb{E}[(N-X_1)\,(X_1)_{(r-1)}\, (X_2+1)_{(s)}]. \notag
\end{align}
Here, we can split
$(N-X_1)\,(X_1)_{(r-1)}
= (N-r+1)\,(X_1)_{(r-1)} - (X_1)_{(r)}$ and
$(X_2+1)_{(s)} = (X_2-s+1\ +s)\,(X_2)_{(s-1)} = (X_2)_{(s)} + s\,(X_2)_{(s-1)}$. Hence, the joint factorial moments $\mu_{(r,s)}$ satisfy the recursive scheme
\begin{align}
p_{00}\,\mu_{(r,s)} &+ p_{01}\, \mu_{(r,s)} + p_{01}\, s\,\mu_{(r,s-1)} = p_{10}\, (N-r+1)\,\mu_{(r-1,s)} - p_{10}\, \mu_{(r,s)} \notag \\
 &
 + p_{11}\, \big((N-r+1)\,\mu_{(r-1,s)} + (N-r+1)\,s\,\mu_{(r-1,s-1)} - \mu_{(r,s)} - s\, \mu_{(r,s-1)}\big). \notag
\end{align}
This can be rewritten as
\ba
\label{facmom_BVB}
\mu_{(r,s)} =  (N-r+1)\,a_1\, \mu_{(r-1,s)} - s\,a_2\, \mu_{(r,s-1)} + (N-r+1)\,s\,a\, \mu_{(r-1,s-1)}.
\ea
Analogously, we could use \eqref{eq:stein2} for deriving a mirror-inverted recursive scheme for $\BVB(N; a_1, a_2, \phi_a)$.

Next, inserting $f(x,y)=(x-1)_{(r-1)}\, y_{(s)}$ into \eqref{eq:BNB-Stein1} (or analogously into \eqref{eq:BNB-Stein2}), we obtain
\begin{align}
\mathbb{E}[(X_1)_{(r)}\,&\, (X_2)_{(s)}] - \pi_2\, \mathbb{E}[(X_1)_{(r)}\, (X_2+1)_{(s)}] \notag \\
 &= \pi_1\, \mathbb{E}[(\nu+X_1)\,(X_1)_{(r-1)}\, (X_2)_{(s)}] + \pi_0\, \mathbb{E}[(\nu+X_1)\,(X_1)_{(r-1)}\, (X_2+1)_{(s)}]. \notag
\end{align}
Here, we can split
$(\nu+X_1)\,(X_1)_{(r-1)}
= (\nu+r-1)\,(X_1)_{(r-1)} + (X_1)_{(r)}$, and
$(X_2+1)_{(s)} = (X_2)_{(s)} + s\,(X_2)_{(s-1)}$ as before. Hence, the joint factorial moments $\mu_{(r,s)}$ satisfy the recursive scheme
\begin{align}
\mu_{(r,s)} &- \pi_2\, \mu_{(r,s)} - \pi_2\, s\,\mu_{(r,s-1)} = \pi_1\, (\nu+r-1)\,\mu_{(r-1,s)} + \pi_1\, \mu_{(r,s)} \notag \\
 &
 + \pi_0\, \big((\nu+r-1)\,\mu_{(r-1,s)} + (\nu+r-1)\,s\,\mu_{(r-1,s-1)} + \mu_{(r,s)} + s\, \mu_{(r,s-1)}\big). \notag
\end{align}
This can be rewritten as
\begin{align}
(1-\pi_\bullet)\,\mu_{(r,s)} =  & (\nu+r-1)\,(\pi_1+\pi_0)\, \mu_{(r-1,s)} \label{facmom_BNB} \\
 &
 + s\,(\pi_2+\pi_0)\, \mu_{(r,s-1)} + (\nu+r-1)\,s\,\pi_0\, \mu_{(r-1,s-1)}. \notag
\end{align}
Both schemes \eqref{facmom_BVB} and \eqref{facmom_BNB} are very useful for practice, as the closed-form expressions for the factorial moments of $\BVB(N; a_1, a_2, \phi_a)$ and $\BNB(\nu;\pi_{1},\pi_{2},\pi_{0})$, as given in \citet{kocherlakota14}, are quite cumbersome. Generally, it appears to be a promising task for future research to derive further moment formulae for $\BVB(N; a_1, a_2, \phi_a)$ and $\BNB(\nu;\pi_{1},\pi_{2},\pi_{0})$, e.g., expressions for inverse or partial moments, like it was done by \citet{weiss2022} in case of the bivariate Poisson distribution.

\subsection{Goodness-of-fit Testing}
\label{Goodness-of-fit Testing}
As already surveyed in Section~\ref{Introduction}, univariate Stein identities have been widely used for constructing goodness-of-fit (GoF) tests, see \citet{betsch22,weiss2023,aleksandrov22,aleksandrov24} for approaches and references. GoF-testing of bivariate distributions is more demanding due to the increased complexity of the probability models, but Stein characterizations again offer the potential for constructing appropriate test statistics. This shall be briefly illustrated in the sequel for the case of testing a bivariate-Poisson null hypothesis, while more comprehensive analyses as well as different null hypotheses are recommended for future research. Inspired by existing works on univariate Stein-GoF tests for Poissonity, we propose to derive a test statistic from \eqref{thm-eq-3}, namely
\ba
\label{BPoi_GoF}
T_{1;f} := \frac{(\mu_1-\sqrt{\mu_1\mu_2}\,\rho)\, \E[f(X_1+1,X_2)] - (\mu_2-\sqrt{\mu_1\mu_2}\,\rho)\, \E[f(X_1,X_2+1)]}{\E[(X_1-X_2)\, f(X_1,X_2)]},
\ea
where $\E[(X_1-X_2)\, f(X_1,X_2)] \not=0$ is required, and where $\lambda_1,\lambda_2$ in \eqref{thm-eq-3} were substituted by the moment expressions $\lambda_i=\mu_i-\sqrt{\mu_1\mu_2}\,\rho$ for $i=1,2$ with $\mu_i=\E[X_i]$ and $\rho=\operatorname{Corr}(X_1,X_2)$. Here, we followed the recommendation of \citet{best97} to express~$\lambda_0$ as $\sqrt{\mu_1\mu_2}\,\rho$ instead of $\operatorname{Cov}(X_1,X_2)$. The sample version~$\hat{T}_{1;f}$ of statistic \eqref{BPoi_GoF} follows by substituting all population moments by corresponding sample moments. If the null hypothesis $\fX = (X_1,X_2)^\top\sim \BPoi(\lambda_0; \lambda_1,\lambda_2)$ is satisfied, then $T_{1;f}=1$ for any choice of~$f$, so deviations of~$\hat{T}_{1;f}$ from~1 may indicate a violation of the BPoi-null. 

\begin{rmk}
\label{rmk_omnibus}
Note that focusing on a particular choice of~$f$ does not lead to a consistent test as the Stein characterization requires $T_{1;f}=1$ to hold for all possible choices of~$f$. Instead, the $\hat{T}_{1;f}$-test follows a similar intuition as the popular dispersion index test, i.e., where the test focuses on particular types of alternatives (``targeted diagnosis''). If one wants to obtain an omnibus test that is powerful against a large class of alternatives, a possible solution is to define an integral test statistic in analogy to Section~4.4 in \citet{aleksandrov24}. But for the sake of illustration, let us focus targeted $\hat{T}_{1;f}$-tests in the sequel, whereas the development and investigation of omnibus tests is recommended for future research.
\end{rmk}
Before looking at some simulation results, it is interesting to study the special case $f(x,y)=x-y$. Then, denoting $\sigma_i^2=\V[X_i]$ for $i=1,2$, $T_{1;f}$ becomes
\begin{align*}
T_{1;x-y}
= & \frac{(\mu_1-\sqrt{\mu_1\mu_2}\,\rho)\, (\mu_1-\mu_2+1) - (\mu_2-\sqrt{\mu_1\mu_2}\,\rho)\, (\mu_1-\mu_2-1)}{\sigma_1^2+\mu_1^2 - 2 (\sqrt{\mu_1\mu_2}\,\rho+\mu_1\mu_2) + \sigma_2^2+\mu_2^2}
\\
= & \frac{\mu_1 + \mu_2 +(\mu_1-\mu_2)^2 -2\sqrt{\mu_1\mu_2}\,\rho}{\sigma_1^2 + \sigma_2^2+(\mu_1-\mu_2)^2 - 2 \sqrt{\mu_1\mu_2}\,\rho}.
\end{align*}
This shows that the special case $f(x,y)=x-y$ leads to a kind of bivariate dispersion index, similar to those in \citet{crockett79}, \citet{loukas86}, \citet{best97}, and \citet{weiss2019}, for example. Conversely, this implies that the ``bivariate Stein index'' \eqref{BPoi_GoF} constitutes a flexible extension of existing bivariate dispersion indexes.

\bigskip
While more comprehensive analyses are planned for future research, we did an initial simulation study in order to investigate the potential of BPoi-GoF tests based on a bivariate Stein index according to \eqref{BPoi_GoF}. We considered two instances of the ``weight function''~$f$ involved in \eqref{BPoi_GoF}, namely the alternating functions $f_a(x,y)=x^a-y^a$ with $a\in\{1/2, 1\}$. As competitors, we initially considered the ``classical'' dispersion indexes investigated by \citet{best97}. However, in agreement with their analyses, we recognized severe size distortions for the tests related to the index of \citet{loukas86} (denoted by $I_B$ and $I_B^*$ in \citet{best97}). Thus, we finally decided to focus on the recommended dispersion-index test of \citet{best97}, namely
\ba
\label{Best}
T^*\ =\ \frac{\mu_2^2(\sigma_1^2-\mu_1)^2 + \mu_1^2(\sigma_2^2-\mu_2)^2 - 2\mu_1\mu_2\,(\sigma_1^2-\mu_1)\,(\sigma_2^2-\mu_2)\,\rho^2}{2\mu_1^2\mu_2^2(1-\rho^4)},
\ea
which constitutes a simplified type of Crockett's index \citep{crockett79}, and where we corrected a misprint in the denominator on p.~427 of \citet{best97}. The dispersion-index test uses the sample counterpart~$\hat{T}^*$ as the test statistic, computed from a bivariate count data set $\fX_1,\ldots,\fX_n$. Its asymptotic distribution satisfies $n\,\hat{T}^* \sim \chi_2^2$ under the null hypothesis of BPoi-counts.

In our simulation study, we considered all null and alternative scenarios like in \citet{best97} with sample sizes $n\in\{50, 100, 200, 500\}$, complemented by further alternatives from the BNB- and BVB-family. More precisely, the considered distributions are
\ba
\label{DistSim}
\begin{array}{@{}l}
\text{``BPoi-1''} = \BPoi(0.1;1.25,0.8),
\quad
\text{``BPoi-2''} = \BPoi(1;5,5),
\\
\text{``BPoi-3''} = \BPoi(0.1;0.2,0.3),
\quad
\text{``BPoi-4''} = \BPoi(1;2.5,2.25),
\\
\text{``BPoi-5''} = \BPoi(1;1,1),
\quad
\text{``BPoi-6''} = \BPoi(0.8;0.2,0.3),
\\
\text{``BPoi-7''} = \BPoi(4;1,1),
\\
\text{``BHerm-1''} = \BHerm(0.75,0.25,0.5,0.15,0.1),
\\
\text{``BHerm-2''} = \BHerm(1,0.75,1.25,0.5,1),
\\
\text{``BHerm-3''} = \BHerm(2,1.5,2,1.5,1),
\\
\text{``BVB-1''} = \BVB(10;0.35,0.325,0.3),
\quad
\text{``BVB-2''} = \BVB(10;0.2,0.2,0.5),
\\
\text{``BNB-1''} = \BNB(9.5;0.2,0.19,0.02),
\quad
\text{``BNB-2''} = \BNB(5;0.2,0.2,0.05),
\end{array}\hspace{-3ex}
\ea
where BPoi-1 to BPoi-7 satisfy the BPoi-null hypothesis, whereas the other scenarios are alternatives. Here, ``BHerm'' refers to a bivariate Hermite distribution with parameters $\lambda_1,\ldots,\lambda_5>0$, which is the distribution of $(Z_1+2\,Z_2+Z_5, Z_3+2\,Z_4+Z_5)$ with mutually independent $Z_i\sim\Poi(\lambda_i)$.
The rejection rates (size or power, respectively) were determined based on 10,000 replications per scenario, where we used a parametric bootstrap approach for implementing the Stein-index tests. Here, to reduce the computational burden of the simulations, the bootstrap simulations are implemented according to the warp-speed method of \citet{giacomini13}.

\begin{table}[th]
\centering\small
\caption{Simulated rejection rates of Stein-index tests $T_{1;f_a}$ with $f_a(x,y)=x^a-y^a$ and $a\in\{1/2, 1\}$ as well as $\hat{T}^*$-test for distributions \eqref{DistSim}, see Section~\ref{Goodness-of-fit Testing} for details.}
\label{tab_Stein_GoF_BPoi}

\smallskip
\begin{tabular}{rl|ccc|l|ccc}
\toprule
\multicolumn{1}{l}{$n$} & Distrib. & $T^*$ & $T_{1;f_1}$ & $T_{1;f_{1/2}}$ & Distrib. & $T^*$ & $T_{1;f_1}$ & $T_{1;f_{1/2}}$ \\
\midrule
50 & BPoi-1 & 0.054 & 0.051 & 0.046 & BHerm-1 & 0.555 & 0.538 & 0.416 \\
100 &  & 0.053 & 0.047 & 0.049 &  & 0.820 & 0.831 & 0.729 \\
200 &  & 0.053 & 0.049 & 0.049 &  & 0.979 & 0.985 & 0.958 \\
500 &  & 0.054 & 0.053 & 0.052 &  & 1.000 & 1.000 & 1.000 \\
\midrule
50 & BPoi-2 & 0.054 & 0.051 & 0.051 & BHerm-2 & 0.595 & 0.574 & 0.471 \\
100 &  & 0.053 & 0.050 & 0.052 &  & 0.851 & 0.872 & 0.793 \\
200 &  & 0.050 & 0.056 & 0.056 &  & 0.986 & 0.992 & 0.973 \\
500 &  & 0.050 & 0.052 & 0.053 &  & 1.000 & 1.000 & 1.000 \\
\midrule
50 & BPoi-3 & 0.052 & 0.054 & 0.045 & BHerm-3 & 0.815 & 0.819 & 0.810 \\
100 &  & 0.045 & 0.052 & 0.046 &  & 0.971 & 0.977 & 0.976 \\
200 &  & 0.049 & 0.054 & 0.052 &  & 1.000 & 1.000 & 1.000 \\
500 &  & 0.046 & 0.050 & 0.049 &  & 1.000 & 1.000 & 1.000 \\
\midrule
50 & BPoi-4 & 0.052 & 0.052 & 0.051 & BVB-1 & 0.463 & 0.768 & 0.709 \\
100 &  & 0.051 & 0.052 & 0.052 &  & 0.906 & 0.980 & 0.950 \\
200 &  & 0.053 & 0.055 & 0.050 &  & 0.999 & 1.000 & 1.000 \\
500 &  & 0.047 & 0.051 & 0.049 &  & 1.000 & 1.000 & 1.000 \\
\midrule
50 & BPoi-5 & 0.054 & 0.053 & 0.053 & BVB-2 & 0.088 & 0.259 & 0.214 \\
100 &  & 0.052 & 0.050 & 0.053 &  & 0.261 & 0.497 & 0.366 \\
200 &  & 0.053 & 0.046 & 0.052 &  & 0.610 & 0.822 & 0.629 \\
500 &  & 0.050 & 0.050 & 0.047 &  & 0.982 & 0.996 & 0.947 \\
\midrule
50 & BPoi-6 & 0.061 & 0.053 & 0.051 & BNB-1 & 0.570 & 0.542 & 0.444 \\
100 &  & 0.054 & 0.049 & 0.050 &  & 0.810 & 0.826 & 0.736 \\
200 &  & 0.051 & 0.051 & 0.051 &  & 0.972 & 0.977 & 0.950 \\
500 &  & 0.050 & 0.044 & 0.049 &  & 1.000 & 1.000 & 1.000 \\
\midrule
50 & BPoi-7 & 0.058 & 0.057 & 0.055 & BNB-2 & 0.676 & 0.644 & 0.455 \\
100 &  & 0.055 & 0.054 & 0.050 &  & 0.899 & 0.904 & 0.751 \\
200 &  & 0.051 & 0.046 & 0.047 &  & 0.994 & 0.995 & 0.961 \\
500 &  & 0.051 & 0.056 & 0.055 &  & 1.000 & 1.000 & 1.000 \\
\bottomrule
\end{tabular}
\end{table}

The obtained simulation results are summarized in Table~\ref{tab_Stein_GoF_BPoi}. The rejection rates in the left part are simulated sizes, which are reasonably close to the nominal level of~5\,\%. Hence, any of the considered tests holds the level to a satisfactory extent. The rejection rates in the right part, in turn, are simulated power values, which should be as large as possible. For the BHerm- and BNB-alternatives, which are characterized by overdispersion (i.e., their variances exceed the respective means), we recognize that~$T^*$ is most powerful for $n=50$, but~$T_{1;f_1}$ is otherwise. For the BVB-alternatives (exhibiting underdispersion), the $T_{1;f_1}$-test is even most powerful without exception, and especially much more powerful than the $T^*$-test. So altogether, the Stein-index approach shows an appealing finite-sample performance, although we only investigated two simple choices for~$f$ so far. For future research, it is recommended to extend the performance analyses to different choices of~$f$ (and omnibus tests) as well, and to find optimal choices for~$f$ in analogy to the analyses in \citet{weiss2023}.

\subsection{Testing for Symmetry (Exchangeability)}
\label{Testing for Symmetry}
In analogy to the general GoF-testing discussed in Section~\ref{Goodness-of-fit Testing}, bivariate Stein characterizations may also be utilized for the more specific task of testing for the symmetry of the considered bivariate distribution, i.e., if $p(x,y)=p(y,x)$ for all~$x,y$. Equivalently, the symmetry null requires that $(X_1,X_2)$ has the same distribution as $(X_2,X_1)$, i.e., $X_1$ and~$X_2$ are exchangeable.

Let us start again with the case $(X_1,X_2)\sim\BPoi(\lambda_0;\lambda_1,\lambda_2)$, then the additional requirement of symmetry (exchangeability) is equivalent to requiring $\lambda_1=\lambda_2 =:\lambda$. Consider the moment $\E[f(X_1+1,X_2+1)]$, which appears in both equations \eqref{thm-eq-1} and \eqref{thm-eq-2}.
If~$f$ is an alternating function, i.e., if $f(x,y)=-f(y,x)$ (like the examples considered in Section~\ref{Goodness-of-fit Testing}), then
$\E[f(X_1+1,X_2+1)] = -\E[f(X_2+1,X_1+1)]$. Under the \emph{additional} exchangeability of $(X_1,X_2)$, the last moment satisfies $\E[f(X_2+1,X_1+1)]=\E[f(X_1+1,X_2+1)]$, i.e., altogether, we obtain
$$
\E[f(X_1+1,X_2+1)] = -\E[f(X_2+1,X_1+1)] = -\E[f(X_1+1,X_2+1)],
$$
so $\E[f(X_1+1,X_2+1)]=0$ follows. Hence, if $(X_1,X_2)\sim\BPoi(\lambda_0;\lambda,\lambda)$, then the original identities \eqref{thm-eq-1} and \eqref{thm-eq-2} reduce to
  \begin{align*}
    \E[X_1\, f(X_1,X_2)]=&\lambda\, \E[f(X_1+1,X_2)],
  \\
    \E[X_2\, f(X_1,X_2)]=&\lambda\, \E[f(X_1,X_2+1)].
  \end{align*}
This suggests defining the test statistic
\ba
\label{SymmTest_BPoi}
T_{2;f} := \Big|\E[X_1\, f(X_1,X_2)]-\lambda\, \E[f(X_1+1,X_2)]\Big| + \Big|\E[X_2\, f(X_1,X_2)]-\lambda\, \E[f(X_1,X_2+1)]\Big|
\ea
with alternating function~$f$. $\lambda$ may be substituted by the moment expression $\mu(1-\rho)$ with $\mu=\frac{1}{2}(\mu_1+\mu_2)$,
and the sample version~$\hat{T}_{2;f}$ of statistic \eqref{SymmTest_BPoi} follows again by substituting all population moments by corresponding sample moments.
If $(X_1,X_2)\sim\BPoi(\lambda_0;\lambda,\lambda)$, then $T_{2;f}=0$, so $\hat{T}_{2;f}>0$ indicates a possible violation of the $\BPoi(\lambda_0;\lambda,\lambda)$-null.

Instead of this specific symmetry hypothesis (symmetry within the BPoi-distribution), one may also more generally consider the null of a symmetric bivariate distribution (this would also cover, e.g., a symmetric BVB- or BNB-distribution). Then, we can just look at the moments $\E[f(X_1+1,X_2)]$ and $\E[f(X_1,X_2+1)]$. If~$f$ is alternating, and if $(X_1,X_2)$ are exchangeable, then
$$
\E[f(X_1+1,X_2)] = -\E[f(X_2,X_1+1)] = -\E[f(X_1,X_2+1)],
$$
i.e., the test statistic
\ba
\label{SymmTest_gen}
T_{3;f} := \E[f(X_1+1,X_2)] + \E[f(X_1,X_2+1)]
\ea
with alternating function~$f$ becomes zero for $(X_1,X_2)$ being exchangeable. For the sample version~$\hat{T}_{3;f}$ of statistic \eqref{SymmTest_gen}, we again substitute all population moments by corresponding sample moments, and deviations of~$\hat{T}_{3;f}$ from~0 indicate a possible violation of the (general) symmetry-null. Note, however, that although the test statistic~$T_{3;f}$ itself is not restricted to a certain distribution family, the implementation of the test (either based on asymptotic derivations or a bootstrap approach) typically requires parametric assumptions.

\begin{rmk}
\label{rmkTimeSeries}
The case of a symmetric bivariate count distribution is particularly relevant in time series analysis.
Let $(X_t)_{t\in\Z=\{\ldots,-1,0,1,\ldots\}}$ be a stationary count process with mean~$\mu>0$ and autocorrelation function $\rho(h)=\operatorname{Corr}(X_t,X_{t-h})$ with time lag $h\in\N$. There is a large family of so-called Poisson integer-valued autoregressive moving-average (INARMA) models, for which pairs $(X_t,X_{t-h})$ are distributed according to $\BPoi\big(\rho(h)\,\mu;\ (1-\rho(h))\,\mu,\ (1-\rho(h))\,\mu\big)$, see Section~3.1 in \citet{aleksandrov24}. Thus, together with \eqref{thm-eq-3}, it follows that the equality
  \begin{equation}
	\label{PoiINARMA_statistic}
    \frac{\E\big[(X_t-X_{t-h})\, f(X_t,X_{t-h})\big]}{\E\big[f(X_t+1,X_{t-h}) - f(X_t,X_{t-h}+1)\big]} = \big(1-\rho(h)\big)\,\mu
  \end{equation}
has to hold, irrespective of the choice of~$f$, as long as the expectation in the denominator is non-zero. Note that the latter excludes the use of any symmetric function~$f$, only such~$f$ satisfying $f(x,y)\not=f(y,x)$ are allowed within \eqref{PoiINARMA_statistic}. It should be noted that analogous model families with symmetric BVB- or BNB-distributions for the pairs $(X_t,X_{t-h})$ exist, see \citet{aleksandrov22} and \citet{aleksandrov24} for details.
Future research should investigate if relations such as \eqref{PoiINARMA_statistic} could be useful for model diagnostics in count time series analysis, in analogy to \citet{weiss2019}.
\end{rmk}

\begin{table}[th]
\centering\small
\caption{Simulated rejection rates of symmetry tests~$T_{2;f_a}$ and~$T_{3;f_a}$ with $f_a(x,y)=x^a-y^a$ and $a\in\{1/2, 1\}$ for distributions \eqref{DistSim}, see Section~\ref{Testing for Symmetry} for details.}
\label{tab_Stein_Sym_BPoi}

\smallskip
\begin{tabular}{rl|cccc|l|cccc}
\toprule
\multicolumn{1}{l}{$n$} & Distrib. & $T_{2;f_1}$ & $T_{2;f_{1/2}}$ & $T_{3;f_1}$ & $T_{3;f_{1/2}}$ & Distrib. & $T_{2;f_1}$ & $T_{2;f_{1/2}}$ & $T_{3;f_1}$ & $T_{3;f_{1/2}}$ \\
\midrule
50 & BPoi-1 & 0.225 & 0.371 & 0.579 & 0.558 & BHerm-1 & 0.591 & 0.615 & 0.566 & 0.545 \\
100 &  & 0.419 & 0.650 & 0.874 & 0.871 &  & 0.886 & 0.883 & 0.833 & 0.808 \\
200 &  & 0.720 & 0.916 & 0.991 & 0.989 &  & 0.996 & 0.995 & 0.982 & 0.976 \\
500 &  & 0.985 & 1.000 & 1.000 & 1.000 &  & 1.000 & 1.000 & 1.000 & 1.000 \\
\midrule
50 & BPoi-2 & 0.050 & 0.052 & 0.049 & 0.055 & BHerm-2 & 0.385 & 0.379 & 0.144 & 0.135 \\
100 &  & 0.046 & 0.045 & 0.050 & 0.051 &  & 0.669 & 0.643 & 0.230 & 0.195 \\
200 &  & 0.049 & 0.046 & 0.042 & 0.043 &  & 0.948 & 0.922 & 0.359 & 0.287 \\
500 &  & 0.054 & 0.053 & 0.050 & 0.053 &  & 1.000 & 1.000 & 0.645 & 0.527 \\
\midrule
50 & BPoi-3 & 0.106 & 0.153 & 0.136 & 0.155 & BHerm-3 & 0.615 & 0.657 & 0.102 & 0.117 \\
100 &  & 0.178 & 0.264 & 0.262 & 0.279 &  & 0.918 & 0.924 & 0.106 & 0.113 \\
200 &  & 0.341 & 0.488 & 0.521 & 0.500 &  & 0.999 & 0.999 & 0.108 & 0.117 \\
500 &  & 0.701 & 0.859 & 0.881 & 0.864 &  & 1.000 & 1.000 & 0.113 & 0.122 \\
\midrule
50 & BPoi-4 & 0.083 & 0.093 & 0.121 & 0.120 & BVB-1 & 0.069 & 0.101 & 0.071 & 0.068 \\
100 &  & 0.128 & 0.143 & 0.196 & 0.199 &  & 0.263 & 0.350 & 0.158 & 0.159 \\
200 &  & 0.204 & 0.234 & 0.368 & 0.343 &  & 0.836 & 0.850 & 0.332 & 0.322 \\
500 &  & 0.439 & 0.509 & 0.732 & 0.706 &  & 1.000 & 1.000 & 0.767 & 0.747 \\
\midrule
50 & BPoi-5 & 0.057 & 0.055 & 0.048 & 0.050 & BVB-2 & 0.016 & 0.026 & 0.020 & 0.027 \\
100 &  & 0.046 & 0.048 & 0.043 & 0.046 &  & 0.017 & 0.025 & 0.026 & 0.030 \\
200 &  & 0.047 & 0.050 & 0.049 & 0.049 &  & 0.023 & 0.042 & 0.027 & 0.029 \\
500 &  & 0.045 & 0.044 & 0.043 & 0.043 &  & 0.183 & 0.240 & 0.029 & 0.034 \\
\midrule
50 & BPoi-6 & 0.124 & 0.164 & 0.136 & 0.150 & BNB-1 & 0.265 & 0.272 & 0.127 & 0.132 \\
100 &  & 0.197 & 0.260 & 0.261 & 0.265 &  & 0.456 & 0.466 & 0.161 & 0.165 \\
200 &  & 0.348 & 0.469 & 0.512 & 0.476 &  & 0.807 & 0.818 & 0.233 & 0.230 \\
500 &  & 0.718 & 0.856 & 0.881 & 0.845 &  & 0.999 & 0.998 & 0.415 & 0.392 \\
\midrule
50 & BPoi-7 & 0.047 & 0.047 & 0.038 & 0.041 & BNB-2 & 0.271 & 0.250 & 0.096 & 0.094 \\
100 &  & 0.051 & 0.052 & 0.044 & 0.048 &  & 0.479 & 0.461 & 0.099 & 0.096 \\
200 &  & 0.054 & 0.054 & 0.046 & 0.052 &  & 0.820 & 0.770 & 0.100 & 0.094 \\
500 &  & 0.053 & 0.053 & 0.048 & 0.049 &  & 1.000 & 0.998 & 0.101 & 0.095 \\
\bottomrule
\end{tabular}
\end{table}

Like for the GoF-tests in Section~\ref{Goodness-of-fit Testing}, we did a simulation study with 10,000 replications per scenario for investigating the finite-sample performance of both the~$T_{2;f}$- and $T_{3;f}$-test, where we consider again the weight functions $f_a(x,y)=x^a-y^a$ with $a\in\{1/2, 1\}$. We implement the tests by a parametric bootstrap approach assuming the null hypothesis of $\BPoi(\lambda_0,\lambda,\lambda)$-counts, i.e., a symmetric BPoi-distribution. The simulations use again the warp-speed method of \citet{giacomini13}, and the considered distributions are those of equation \eqref{DistSim} with sample sizes $n\in\{50, 100, 200, 500\}$ like before. However, by contrast to Section~\ref{Goodness-of-fit Testing}, the $\BPoi(\lambda_0,\lambda,\lambda)$-null is now satisfied only for scenarios ``BPoi-2'', ``BPoi-5'', and ``BPoi-7''. The alternatives ``BHerm-3'', ``BVB-2'', and ``BNB-2'', in turn, are particularly demanding, because although these count distributions are symmetric, they are no BPoi-distributions and hence violate the $\BPoi(\lambda_0,\lambda,\lambda)$-null anyway.

The simulation results in Table~\ref{tab_Stein_Sym_BPoi} show that again, all tests have empirical sizes (the rejection rates for scenarios ``BPoi-2'', ``BPoi-5'', and ``BPoi-7'') being reasonably close to the nominal 5\,\%-level. Regarding the power concerning asymmetry \emph{within} the BPoi-family (remaining BPoi-scenarios), we have increasing power for increasing sample size~$n$, where the best power is usually achieved by one of the $T_{3;f}$-tests, closely followed by the $T_{2;f_{1/2}}$-test. So the less specific $T_{3;f}$-statistic appears to be the preferable solution here. Things change, however, if the alternative is \emph{outside} the BPoi-family. If the alternative is not BPoi-distributed but still symmetric (scenarios ``BHerm-3'', ``BVB-2'', and ``BNB-2''), then the $T_{3;f}$-test is hardly powerful at all, which is plausible as its statistic solely focuses on symmetry. In fact, we conjecture that it might be possible in future research to develop a non-parametric implementation of the $T_{3;f}$-test, which allows one to test a symmetry-null irrespective of the apparent distribution family. For the remaining alternatives, where we are concerned with both non-BPoi and asymmetric distributions, the $T_{3;f}$-tests show improved power but still, they are clearly outperformed by any of the $T_{2;f}$-tests. Again, this is plausible as the $T_{2;f}$-statistic aggregates both the distributional deviations and the asymmetry. Here, it is interesting to note that in most cases, both instances~$T_{2;f_{1}}$ and~$T_{2;f_{1/2}}$ perform virtually the same. But for the BVB-alternatives being characterized by underdispersion, the weight function~$f_{1/2}$ appears to be beneficial. Altogether, the proposed symmetry tests constitute a promising approach for uncovering asymmetry, where the weight function~$f_{1/2}$ either leads to similar or even better results than~$f_1$. The latter is opposite to the GoF-tests in Section~\ref{Goodness-of-fit Testing}, where~$f_1$ is the preferable choice. In both cases, an in-depth analysis on the optimal choice of~$f$ appears to be a challenging task for future research.

\subsection{Real-World Data Analyses}
\label{Real-World Data Analyses}
We conclude the applications section by some real-world data examples. For comparability to the existing literature, we reanalyze all data examples presented by \citet{crockett79} (also see \citet{loukas86}) and \citet{best97}. These are
\begin{itemize}
	\item Data-1: number of plants of species ``Lacistema aggregatum'' and ``Protium guianense'' in 100 contiguous quadrats \citep[Table~2]{crockett79};
	\item Data-2: number of children suffering medically attended injuries in each of two different time periods, ``ages 4--7'' against ``ages 8--11'' \citep[Table~3]{crockett79};
	\item Data-3: accidents for 122 engine drivers in two consecutive periods, ``1937--1942'' against ``1943--1947'' \citep[Table~5]{best97};
	\item Data-4: synthetic data \citep[Table~6]{best97}, originally used to illustrate problems with the test statistics from \citet{crockett79} and \citet{loukas86}.
\end{itemize}
Important descriptive statistics of these data sets are provided by Table~\ref{tab_data examples}\,(a), where the estimate of $\mu=(\mu_1+\mu_2)/2$ is relevant for symmetry testing.

\begin{table}[th]
\centering\small
\caption{Data examples of Section~\ref{Real-World Data Analyses}: (a) sample size~$n$ and descriptive statistics; (b) test statistics (P-values in parentheses).}
\label{tab_data examples}

\smallskip
(a)~\begin{tabular}{l|rccccc|c}
\toprule
Data & \multicolumn{1}{c}{$n$} & $\mu_1$ & $\mu_2$ & $\sigma_1^2/\mu_1$ & $\sigma_2^2/\mu_2$ & $\rho$ & $\mu$ \\
\midrule
Data-1 & 100 & 0.940 & 0.640 & 1.415 & 1.216 & 0.276 & 0.790 \\
Data-2 & 621 & 1.126 & 1.325 & 1.350 & 1.483 & 0.185 & 1.225 \\
Data-3 & 122 & 1.270 & 0.975 & 1.301 & 1.330 & 0.259 & 1.123 \\
Data-4 & 50 & 0.620 & 0.520 & 1.112 & 1.353 & 0.846 & 0.570 \\
\bottomrule
\end{tabular}

\vspace{2ex}
(b)~\begin{tabular}{l|ccc|cccc}
\toprule
Data & $T^*$ & $T_{1;f_1}$ & $T_{1;f_{1/2}}$ & $T_{2;f_1}$ & $T_{2;f_{1/2}}$ & $T_{3;f_1}$ & $T_{3;f_{1/2}}$ \\
\midrule
Data-1 & 0.103 & 0.766 & 0.860 & 0.677 & 0.355 & 0.600 & 0.273 \\
 & \footnotesize (0.006) & \footnotesize (0.005) & \footnotesize (0.024) & \footnotesize (0.003) & \footnotesize (0.001) & \footnotesize (0.004) & \footnotesize (0.015) \\[1ex]
Data-2 & 0.172 & 0.708 & 0.845 & 0.883 & 0.234 & -0.399 & -0.160 \\
 & \footnotesize (0.000) & \footnotesize (0.000) & \footnotesize (0.000) & \footnotesize (0.000) & \footnotesize (0.000) & \footnotesize (0.001) & \footnotesize (0.003) \\[1ex]
Data-3 & 0.094 & 0.777 & 0.903 & 0.597 & 0.229 & 0.590 & 0.253 \\
 & \footnotesize (0.003) & \footnotesize (0.004) & \footnotesize (0.098) & \footnotesize (0.029) & \footnotesize (0.029) & \footnotesize (0.010) & \footnotesize (0.022) \\[1ex]
Data-4 & 0.083 & 0.859 & 0.770 & 0.082 & 0.099 & 0.200 & 0.141 \\
 & \footnotesize (0.126) & \footnotesize (0.318) & \footnotesize (0.110) & \footnotesize (0.531) & \footnotesize (0.127) & \footnotesize (0.061) & \footnotesize (0.038) \\
\bottomrule
\end{tabular}
\end{table}

The obtained test results, i.e., the values of the test statistics together with the corresponding P-values, are summarized in Table~\ref{tab_data examples}\,(b), where~$T^*$ was again executed based on the $\chi_2^2$-asymptotics, while the remaining tests rely on a parametric BPoi-bootstrap with 10,000 replications. For data sets~1 and~2, we obtain rejections by any of the tests on the 5\%-level, which confirms the earlier conclusions by \citet{crockett79,loukas86} and complements them by the additional results of the symmetry tests. So the data appear to not originate from a BPoi-distribution, in particular not from a symmetric BPoi-distribution, which is also confirmed by the unequal mean estimates as well as the dispersion ratios exceeding~1 in the first two rows of Table~\ref{tab_data examples}\,(a). The situation for data set~3 is analogous, but here, the $T_{1;f_{1/2}}$-test does not lead to a rejection on the 5\,\%-level. This confirms our conclusions from Section~\ref{Goodness-of-fit Testing}, where we recognized that the weight function~$f_{1/2}$ has poorer power than its competitors for small sample sizes~$n$. Finally, for data set~4, already \citet{best97} reported that none of their tests was significant, which is confirmed by all tests except~$T_{3;f}$. Generally, the data set is very small ($n=50$) and the deviations from a (symmetric) BPoi-distribution are rather mild, see the last row in Table~\ref{tab_data examples}\,(a), which explains that the tests are generally not powerful. That's why it is all the more surprising that the $T_{3;f_{1/2}}$-test leads to a rejection anyway, and also the P-value of the $T_{3;f_{1}}$-test is quite close to~0.05. In view of our findings in Sections~\ref{Goodness-of-fit Testing} and~\ref{Testing for Symmetry}, we may conclude that data set~4 originates from an asymmetric BPoi-distribution, where the $T_{3;f}$-test, and especially the $T_{3;f_{1/2}}$-test, turned out to be particularly powerful.

\section{Conclusions}
\label{Conclusions}
In this research, we derived novel Stein characterizations for the three most important types of bivariate count distributions, namely the bivariate Poisson, binomial, and negative-binomial distribution. We also demonstrated that these Stein characterizations offer a broad potential for applications in theory and practice. More precisely, we showed how the Stein identities can be utilized for moment calculations, where the novel recursive schemes for bivariate pmfs and factorial moments of BNB- and BVB-distribution appear to be particularly useful from a computational point of view. We also deduced test statistics for the goodness-of-fit or the symmetry of bivariate count distributions (exchangeability), which showed an appealing performance in simulations and real-world data examples.

Given the limited length of a research article, it was of course not possible for us to elaborate on the above application areas in full detail. Instead, it became clear that our novel Stein characterizations offer the opportunity for various follow-up research projects. Among others, the Stein tests for goodness-of-fit or symmetry deserve a more comprehensive analysis, e.g., with respect to optimal weight functions~$f$ or a non-parametric implementation of the $T_{3;f}$-test. Future research might also try to utilize our bivariate Stein characterizations for the remaining application areas reviewed in Section~\ref{Introduction}, such as model diagnostics for count time series analysis.

\subsubsection*{Acknowledgments}
Shaochen Wang was supported by the National Natural Science
Foundation of China (No.\ 12371266).


\end{document}